\documentclass{article}
\usepackage{graphicx} 

\usepackage[margin=1in]{geometry}
\usepackage[utf8]{inputenc}
\usepackage{bbm, bm}
\usepackage{graphicx}
\usepackage{color, xcolor}
\usepackage{amsmath, amsfonts, amssymb, amsthm, thmtools, mathtools}
\usepackage{algorithm, algpseudocode}
\usepackage{empheq}

\usepackage{singer-macros}

\def\draft{1}
\def\doubleblind{0}
\usepackage[margin=1in]{geometry}
\usepackage[utf8]{inputenc}
\usepackage{amsthm}
\usepackage{amsthm, amsmath, amssymb, mathtools}
\usepackage{bm,bbm}
\usepackage{graphicx}
\usepackage{color,xcolor}
\usepackage{paralist,enumitem}
\usepackage{mathrsfs}
\usepackage[normalem]{ulem}
\usepackage{tabularx}
\usepackage{tikz}
\usepackage{caption,comment}
\usepackage{algorithmicx}
\usepackage{algorithm}
\usepackage{algpseudocode}
\newcounter{algsubstate}

\algnewcommand\algorithmicinput{\textbf{Input:}}
\algnewcommand\Input{\item[\algorithmicinput]}

\algnewcommand\algorithmicoutput{\textbf{Output:}}
\algnewcommand\Output{\item[\algorithmicoutput]}

\algnewcommand\algorithmicgoal{\textbf{Goal:}}
\algnewcommand\Goal{\item[\algorithmicgoal]}

\usepackage[colorlinks=true, allcolors=blue]{hyperref}
\usepackage[capitalize,noabbrev,nameinlink]{cleveref}

\newcommand{\snote}[1]{\ifnum\draft=1\textcolor{orange}{[\textbf{Santhoshini:} #1]}\fi}
\newcommand{\mnote}[1]{\ifnum\draft=1\textcolor{red}{[\textbf{Madhur:} #1]}\fi}
\newcommand{\noah}[1]{\ifnum\draft=1\textcolor{teal}{[\textbf{Noah:} #1]}\fi}

\newcommand{\CSPopt}{\mathsf{opt}^{\mathrm{CSP}}}
\newcommand{\LPopt}{\mathsf{opt}^{\mathrm{LP}}}

\newcommand{\Exp}{\mathop{\mathbb{E}}}
\newcommand{\cP}{\mathcal{P}}

\newcommand{\cD}{\mathcal{D}}
\newcommand{\cF}{\mathcal{F}}

\newcommand{\N}{\mathbb{N}}

\newcommand{\R}{\mathbb{R}}
\newcommand{\Z}{\mathbb{Z}}

\newcommand{\maxF}{\textsc{Max-CSP}(\cF)}
\newcommand{\gapMaxF}[2]{{(#1,#2)\textrm{-}\maxF}}
\newcommand{\gbMaxF}{\gapMaxF{\gamma}{\beta}}

\newcommand{\Cut}{\textsc{Cut}}
\newcommand{\MaxCut}{\textsc{Max-Cut}}
\newcommand{\gapMaxCut}[2]{(#1,#2)\textsc{-Max-Cut}}
\newcommand{\DiCut}{\textsc{DiCut}}
\newcommand{\MaxDiCut}{\textsc{Max-DiCut}}
\newcommand{\gapMaxDiCut}[2]{(#1,#2)\textsc{-Max-DiCut}}
\newcommand{\MaxkAnd}{\textsc{Max-}k\textsc{AND}}
\newcommand{\MaxkMon}{\textsc{Max-}k\textsc{Monarchy}}

\newcommand{\CSPval}{\mathsf{val}^\mathrm{CSP}}

\newcommand{\yes}{\textbf{YES}}
\newcommand{\no}{\textbf{NO}}

\newcommand{\sgyf}{S_\gamma^Y(\cF)}
\newcommand{\sbnf}{S_\beta^N(\cF)}
\newcommand{\kgyf}{K_\gamma^Y(\cF)}
\newcommand{\kbnf}{K_\beta^N(\cF)}

\newcommand{\vecb}{\mathbf{b}}

\newcommand{\vecj}{\mathbf{j}}

\newcommand{\vecmu}{\boldsymbol{\mu}}

\usepackage{amsthm}
\usepackage{thmtools,thm-restate}
\numberwithin{equation}{section}

\usepackage[colorlinks=true, allcolors=blue]{hyperref}
\usepackage[capitalize,noabbrev,nameinlink]{cleveref}

\declaretheoremstyle[bodyfont=\it,qed=\qedsymbol]{noproofstyle}

\declaretheorem[name=Observation,numbered=no]{observation*}

\declaretheorem[numberlike=equation]{theorem}

\declaretheorem[name=Theorem,numbered=no]{theorem*}

\declaretheorem[name=Lemma,numbered=no]{lemma*}

\declaretheorem[name=Corollary,numbered=no]{corollary*}

\declaretheorem[name=Proposition,numbered=no]{proposition*}

\declaretheorem[name=Claim,numbered=no]{claim*}

\declaretheorem[numberlike=equation]{conjecture}
\declaretheorem[name=Conjecture,numbered=no]{conjecture*}

\declaretheorem[name=Question,numbered=no]{question*}

\declaretheoremstyle[bodyfont=\it]{defstyle} 

\declaretheorem[numberlike=equation,style=defstyle]{definition}
\declaretheorem[unnumbered,name=Definition,style=defstyle]{definition*}

\declaretheorem[unnumbered,name=Notation=defstyle]{notation*}

\declaretheorem[unnumbered,name=Construction,style=defstyle]{construction*}

\declaretheoremstyle[
  headfont=\normalfont\bfseries,
  bodyfont=\normalfont,
  postheadspace=1em,
  qed=$\lozenge$
]{rmkstyle}

\declaretheorem[
  style=rmkstyle,
  numberlike=equation
]{remark}

\title{Sketching approximations and LP approximations for finite CSPs are related}
\date{}
\usepackage[style=alphabetic, backend=biber, minalphanames=3, maxalphanames=4, maxbibnames=99, maxcitenames=99]{biblatex}

\def\doubleblind{0}
\def\draft{1}
\ifnum\doubleblind=0
\author{Noah G. Singer\thanks{Carnegie Mellon University, Pittsburgh, PA, USA. Email: \texttt{ngsinger@cs.cmu.edu}. Supported in part by an NSF Graduate Research Fellowship (Award
DGE 2140739).}
\and Madhur Tulsiani\thanks{Toyota Technological Institute, Chicago, Illinois, USA. Email: \texttt{madhurt@ttic.edu}.}
\and Santhoshini Velusamy\thanks{Toyota Technological Institute, Chicago, Illinois, USA. Email: \texttt{santhoshini@ttic.edu}. Work supported in part by NSF award CCF 2348475.}
}
\fi
\ifnum\doubleblind=1
\author{Anonymous authors}
\fi

\begin{document}

\maketitle
\begin{abstract}
    We identify a connection between the approximability of CSPs in two models:
    (i) sublinear space streaming algorithms,
    and (ii) the basic LP relaxation.
    We show that whenever the basic LP admits an integrality gap, there is an $\Omega(\sqrt{n})$-space sketching lower bound. We also show that all existing linear space streaming lower bounds for Max-CSPs can be lifted to integrality gap instances for basic LPs.
    For bounded-degree graphs, by combining the distributed algorithm of \textcite{Yos11} for approximately solving the basic LP with techniques described in \textcite{SSSV25} for simulating a distributed algorithm by a sublinear space streaming algorithm on bounded-degree instances of Max-DICUT,
    it appears that there are sublinear space streaming algorithms implementing the basic LP, for every CSP. 
    
    Based on our results, we conjecture the following dichotomy theorem:
    Whenever the basic LP admits an integrality gap, there is a linear space single-pass streaming lower bound,
    and when the LP is roundable, there is a sublinear space streaming algorithm.
\end{abstract}

\section{Introduction}

In recent years, there has been substantial attention on the ability of \emph{streaming algorithms} to approximate the value of \emph{constraint satisfaction problems (CSPs)}.

\paragraph*{CSPs.}
Consider the following computational problem, which we denote $\maxF$ (see \Cref{sec:csp} below for a more formal definition):
Fix an \emph{alphabet size} $q \in \N$, an \emph{arity} $k \in \N$, and a \emph{predicate family} $\calF \subseteq ([q]^k)^{\{0,1\}}$.
An \emph{instance} $\Phi$ of $\maxF$ on $n$ \emph{variables} $x_1,\ldots,x_n \in [q]$
consists of $m$ \emph{constraints} $C_1,\ldots,C_m$.
Each constraint $C$ applies a predicate $f \in \calF$ to $k$ distinct variables $x_{i_1},\ldots,x_{i_k}$,
and is \emph{satisfied} iff $f(x_{i_1},\ldots,x_{i_k}) = 1$.
The goal of the $\maxF$ problem is to satisfy as many constraints in $\Phi$ as possible, given $\Phi$.
In particular, we write $\CSPopt_\Phi$ for the highest fraction of constraints which may be satisfied.
For $1 \ge \gamma > \beta \ge 0$, the $\gbMaxF$ problem is the computational problem of distinguishing the cases $\CSPopt_\Phi \ge \gamma$ and $\CSPopt_\Phi \le \beta$.
For $0 \le \alpha \le 1$, we say $\maxF$ is \emph{$\alpha$-approximable} in a computational model if $\gapMaxF{\gamma}{\alpha\gamma-\epsilon}$ is tractable for every $0 \le \gamma \le 1$ and $\epsilon > 0$.

\paragraph*{Streaming algorithms.}
In the \emph{plain streaming model} for $\gbMaxF$, the list of constraints $C_1,\ldots,C_m$ in $\Phi$ is presented to an algorithm sequentially.
The streaming algorithm has memory space which is bounded as a function of $n$, the number of variables in $\Phi$,
and it must distinguish between the cases $\CSPopt_\Phi \ge \gamma$ and $\CSPopt_\Phi \le \beta$ with constant advantage.
The interesting regime is where the memory space is $o(n)$, as once $\tilde{O}(n)$ space is allowed, there is a simple distinguisher based on uniformly subsampling constraints (see, e.g.,~\cite{CGS+22-linear-space}).

We call this model \emph{plain} because it admits various interesting modifications.
Three which are particularly relevant are:
\begin{itemize}
\item when the constraints are promised to arrive in \emph{randomly order}~\cite{KKS15,SSSV23-random-ordering,SSSV25},
\item when \emph{multiple passes} are allowed over the list of constraints~\cite{SSSV23-random-ordering,SSSV25,FMW25},
\item and when the algorithm is required to be a \emph{sketching algorithm}, a special type of ``composable'' streaming algorithm~\cite{CGSV24} (see \Cref{sec:streaming dichotomy} below for a definition).
\end{itemize}
In this paper, to keep the discussion focused, we mostly restrict our attention to the plain model.

\paragraph*{Our problem.}
Some problems admit no ``nontrivial'' algorithms in the streaming setting.
The canonical example of such a problem is $\MaxCut$ ($q=k=2$, and $\calF = \{f_{\Cut}\}$ where $f_{\Cut}(b_1,b_2) = 1$ iff $b_1 \ne b_2$),
whose hardness has been studied in~\cite{KK15,KKS15,KKSV17,KK19,FMW25}.
In particular, \textcite{KK19} showed that for every $\epsilon > 0$, no $o(n)$-space streaming algorithm can solve the $\gapMaxCut1{\frac12+\epsilon}$ problem~\cite{KK19}.
This is the best possible hardness result, because \emph{every} $\MaxCut$ instance $\Phi$ has $\CSPopt_\Phi \ge \frac12$.
(This is because a \emph{random} assignment satisfies $\frac12$ of the constraints in expectation.
Hence, the $\gapMaxCut1{\frac12-\epsilon}$ problem is vacuous.)
More generally, for a predicate family $\calF$, define the trivial approximability threshold
\begin{equation}\label{eq:rho}
\rho(\cF) \coloneqq \lim_{n \to \infty} \bracks*{ \inf_{\Psi \textrm{ instance of }\maxF\text{ on }n\text{ variables}} \{\CSPopt_\Psi\} }.
\end{equation}
Analogous impossibility results are known for  $\gapMaxF1{\rho(\calF)+\epsilon}$ for other problems, like $\textsc{Max-}q\textsc{UniqueGames}$ and $\textsc{Max-}q\textsc{Coloring}$~\cite{GT19,SSV24-jour-version,CGSV24,CGS+22-linear-space}.

Other problems do admit nontrivial algorithms.
For instance, $\MaxDiCut$ ($q=k=2$, and $\calF = \{f_{\DiCut}\}$ where $f_{\DiCut}(b_1,b_2) = 1$ iff $b_1 = 1$ and $b_2 = 0$)
has been studied extensively in~\cite{GVV17,CGV20,SSSV23-dicut,SSSV23-random-ordering,SSSV25}.
The trivial threshold for $\MaxDiCut$ is $\rho(\{f_{\DiCut}\}) = \frac14$,
but $\MaxDiCut$ is:
\begin{itemize}
    \item $(\frac49-\epsilon)$-approximable via $O(\log n)$-space algorithms~\cite{CGV20}
    (and $(\frac49+\epsilon)$-approximations require $\Omega(\sqrt n)$ space),
    \item $0.485$-approximable via $\tilde{O}(\sqrt n)$-space algorithms~\cite{SSSV23-dicut},
    \item and (on bounded-degree instances) $(\frac12-\epsilon)$-approximable via $o_\epsilon(n)$-space algorithms~\cite{SSSV25}.
\end{itemize}
($\frac12$ is a natural threshold here, because $\MaxCut$ lower bounds of~\cite{KK19} also imply that $\gapMaxDiCut1{\frac12+\epsilon}$ requires $\Omega(n)$-space for every $\epsilon > 0$.)
Analogous results are known for other problems, in particular $\MaxkAnd$~\cite{Sin23-kand}.

The motivating question for this paper is to find a unifying explanation for these upper and lower-bounds:
\begin{quote}
    \emph{Why do some CSPs admit nontrivial streaming algorithms in the $o(n)$-space (single-pass) setting, and others do not?}
\end{quote}
In this paper, we conjecture an explanation: The $o(n)$-space streaming approximability of a CSP is intimately tied to its approximability via linear programming (LP) relaxations,
and in particular via a certain ``canonical'' LP relaxation, called the \emph{basic LP}, which is essentially the optimal polynomial-sized relaxation~\cite{GT17,CLRS16,KMR22}.

In the weaker setting of $o(\sqrt n)$-space sketching algorithms, \textcite{CGSV24} obtained a complete, and even efficiently decidable, characterization (see \Cref{thm:cgsv} below for a formal statement).
Their characterization recovers the results of~\cite{CGV20} for $\MaxDiCut$, but cannot explain the advantages obtained in $\tilde{O}(\sqrt n)$ space by~\cite{SSSV23-dicut}.

In the next subsection, we first state our conjecture formally, and then describe the partial evidence that we have obtained towards it.

\subsection{Statement of the conjecture}

We defer formally defining the basic LP until \Cref{sec:basic lp} below.
For the purposes of this introduction, we only need that it is a \emph{relaxation}:
For every instance $\Phi$ of $\maxF$, the basic LP assigns a relaxed value $\LPopt_\Phi \in [0,1]$ which always satisfies $\LPopt_\Phi \ge \CSPopt_\Phi$.
We then make the following definition:

\begin{definition}[Integrality gap instance for basic LP]\label{def:integrality gap instance}
    Let $\calF$ be a predicate family and $1 \ge \gamma > \beta \ge 0$.
    A \emph{$(\gamma,\beta)$-integrality gap instance} for the basic LP for $\maxF$ is an instance of $\Phi$ of $\maxF$ s.t.
    \begin{enumerate}
        \item \emph{Soundness:} $\CSPopt_\Phi \le \beta$.
        \item \emph{Completeness:} $\LPopt_\Phi \ge \gamma$.
    \end{enumerate}
    We say the basic LP for $\maxF$ is \emph{$(\gamma,\beta)$-roundable} if there does \emph{not} exist a $(\gamma,\beta)$-integrality instance,
    i.e., every instance $\Phi$ of $\maxF$ with $\CSPopt_\Phi \le \beta$ has $\LPopt_\Phi < \gamma$.
\end{definition}

Our conjecture is then the following:

\begin{conjecture}[Dichotomy theorem for single-pass streaming algorithms]\label{conj}
    Let $\calF \subseteq \{f : [q]^k \to \{0,1\}\}$.
    For every $1 \ge \gamma > \beta \ge 0$:
    \begin{enumerate}
        \item If there exists a $(\gamma,\beta)$-integrality gap instance for the basic LP for $\maxF$,
        then for every $\epsilon > 0$, every streaming algorithm for $(\gamma-\epsilon,\beta+\epsilon)$-$\maxF$ uses $\Omega(n)$ space. \label{conj:lb}
        \item Otherwise, (the basic LP for $\maxF$ is $(\gamma,\beta)$-roundable, and) for every $\epsilon > 0$,
        there is a $o_\epsilon(n)$-space sketching algorithm which solves $(\gamma+\epsilon,\beta-\epsilon)$-$\maxF$. \label{conj:conv}
    \end{enumerate}
\end{conjecture}

We remark that \Cref{conj:conv} of \Cref{conj} is already known in the special case where $\Phi$ has constant maximum degree $D$
(i.e., every variable occurs in at most $D$ constraints).
The proof of this follows from combining the distributed algorithm of Yoshida~\cite{Yos11} for approximately solving the basic LP
with the connection between distributed algorithms and sublinear space algorithms on bounded-degree instances described in~\cite{SSSV25} (see also \Cref{rmk:alg} below).
(We defer a full proof of this fact to the final version of this paper.)
The challenge in fully proving \Cref{conj:conv} is then removing the bounded-degree assumption;
for the very special case of $0$-round distributed algorithms for \textsc{Max-DiCut}, this was done in~\cite{SSSV23-dicut}.

\begin{remark}\label{rmk:alg}
    At the initial time of writing, we were not aware of the work of Yoshida~\cite{Yos11},
    which was used in the work of~\cite{FMW25-dichotomy} on multi-pass approximations.
    We were aware of the special case of $\MaxkAnd$, where a packing LP equivalent to the basic LP is presented in \cite{Tre98-alg}.
\end{remark}

\begin{remark}
    \textcite{LO17} implement the basic \emph{semidefinite programming} (SDP) relaxation algorithm in polylogarithmic space and polynomial time;
    their work builds on the quasilinear-time algorithm due to \textcite{Ste10}.
    This seems to lead to polylogarithmic-space, polynomial-pass algorithms for implementing the basic SDP.
    Since the basic SDP for $\MaxCut$ achieves a nontrivial approximation (indeed, the ratio is $\approx 0.878 \gg \frac12$)~\cite{GW95},
    $\MaxCut$ does admit a nontrivial approximation in this model.
    This is in contrast to the result of~\textcite{FMW25},
    who show that $\MaxCut$ is approximation-resistant to $p$-pass, $o(n^{1/3}/p)$-space streaming algorithms.
\end{remark}

\subsection{Partial results towards the conjecture}

Our first result shows that (basic) LP approximability is stronger than sketching approximability in $o(\sqrt{n})$ space for every $\maxF$:

\begin{theorem}\label{thm:int gap impl CGSV}
    Let $q, k \in \N$, $\calF \subseteq ([q]^k)^{\{0,1\}}$.
    Suppose there exists a $(\gamma,\beta)$-integrality gap instance for the basic LP for $\maxF$.
    Then for every $\epsilon > 0$,
    every sketching algorithm for $(\gamma-\epsilon,\beta+\epsilon)$-$\maxF$ on $n$-variable instances uses $\Omega(\sqrt n)$ space.
\end{theorem}

\textcite{CGSV24} conjectured that (conditional) inapproximability (i.e., $\mathbf{NP}$- or unique games-hardness) in the polynomial time setting should lead to (unconditional) inapproximability by sketching algorithms.
Our (unconditional) \Cref{thm:int gap impl CGSV} is a stronger version of their conjecture.

Our next two results translate some known streaming lower bounds (from~\cite{CGSV24} and \cite{CGS+22-linear-space})
to the setting of the basic LP.
Firstly, \Cref{thm:cgsv:one-wise} below, due to~\cite{CGSV24}, gives a sufficient condition on $\calF$, called ``weakly supporting one-wise independence'', for streaming approximation-resistance of $\maxF$ in $o(\sqrt n)$ space.

\begin{remark}\label{rmk:monarchy}
While \textcite{CGSV24} show that ``supporting one-wise independence'' is also a necessary condition for approximation resistance of \emph{symmetric} Boolean CSPs,
the same authors, along with Shahrasbi \cite{CGS+22-monarchy} prove that this is not true in general.
In particular, they show that $\MaxkMon$ (for $k \ge 5$) (which does not support one-wise independence) is \emph{sketching} approximation-resistant in $o(\sqrt{n})$ space.
Because $\MaxkMon$ is LP-approximable~\cite{ABM12,ST16}, \Cref{conj} would imply that it is approximable by $o(n)$-space streaming algorithms,
giving the first example of a CSP which is approximation-resistant to $o(\sqrt n)$-space sketching algorithms but not to $o(n)$-space streaming algorithms.
\end{remark}

\begin{remark}\label{rmk:ltf}
If \cref{conj} is true, then this condition would likely also not be tight for approximation resistance in $o(n)$ space, 
since \textcite{Pot19} showed the existence of a CSP which does not support one-wise independence and is inapproximable by its canonical SDP under the unique games conjecture.
\end{remark}

We show that whenever $\calF$ satisfies the one-wise independence condition, it is also approximation-resistant to the basic LP:

\begin{theorem}\label{thm:one-wise}
    Let $k,q \in \N, \cF \subseteq ([q]^k)^{\{0,1\}}$.
    If $\maxF$ fulfills the hypotheses of \Cref{thm:cgsv:one-wise} below,
    $\maxF$ is also approximation-resistant to the basic LP.
\end{theorem}

\Cref{thm:one-wise} is proven in detail as \Cref{thm:one-wise:formal} below.
We also translate certain linear space streaming lower bounds due to~\cite{CGS+22-linear-space}, in the form of \Cref{thm:cgsvv} below,
into lower bounds against the basic LP:

\begin{theorem}\label{thm:wide}
    Let $k,q \in \N, \cF \subseteq (\Z_q^k)^{\{0,1\}}$.
    The $\gbMaxF$ problem shown to be hard in \Cref{thm:cgsvv} below
    is also hard for the basic LP for $\maxF$.
\end{theorem}

\Cref{thm:wide} is proven in detail as \Cref{thm:wide:formal} below.

\subsection*{Bibliographic note}
This is an independent and concurrent work to the recent work of~\cite{FMW25-dichotomy},
who gave a dichotomy theorem showing that multi-pass CSP approximability is equivalent to basic LP approximability.
Our results are independent of theirs, except for that, as mentioned in \Cref{rmk:alg} below,
we were not initially aware of the Yoshida~\cite{Yos11} result giving  a self-contained distributed algorithm for all basic LPs.
The connection between packing LP solvers, distributed algorithms, and CSPs was already suggested in the work of~\cite{SSSV25}.

\section{Background}

\subsection{CSPs}\label{sec:csp}
For positive integers $q$ and $k$, a $q$-ary {\em constraint satisfaction problem} (CSP) is given by a (finite) family of predicates 
$\cF \subseteq ([q]^k)^{\{0,1\}}$. 
A constraint $C$ on $x_1,\ldots,x_n$ is given by a pair $(f,\vecj)$, with $f \in \cF$ and $\vecj = (j_1,\ldots,j_k) \in [n]^k$ where the coordinates of $\vecj$ are all distinct.\footnote{
    To allow repeated variables in a constraint, note that one can turn $\cF$ into $\cF'$ by introducing new functions corresponding to all the possible replications of variables of functions in $\cF$.}
An assignment $\vecb \in [q]^n$ satisfies $C = (f,\vecj)$ if $f(b_{j_1},\ldots,b_{j_k})=1$. 
To every finite set $\cF$, we associate a maximization problem $\maxF$ that is defined as follows: An instance $\Psi$ of $\maxF$ consists of $m$ constraints $C_1,\ldots,C_m$ applied to $n$ variables $x_1,x_2,\dots,x_n$ along with $m$ non-negative integer weights $w_1,\ldots,w_m$. The value of an assignment $\vecb \in [q]^{n}$ on an instance $\Psi = (C_1,\ldots,C_m; w_1,\ldots,w_m)$, denoted $\CSPval_\Psi(\vecb)$, is the fraction of weight of constraints satisfied by $\vecb$. We study the approximability of CSPs through the following gapped promised problems: for $0 \leq \beta < \gamma \leq 1$, the $(\gamma,\beta)$-approximation version of $\maxF$, denoted $\gbMaxF$, is the task of distinguishing between instances from $\Gamma = \{\Psi \mid \CSPopt(\Psi) \geq \gamma\}$ and instances from $B = \{\Psi \mid \CSPopt(\Psi) \leq \beta\}$. 

\subsection{Dichotomies in the streaming setting}\label{sec:streaming dichotomy}

In the streaming setting, an instance $\Psi = (C_1,\ldots,C_m)$ is presented as a stream $\sigma_1,\sigma_2,\ldots,\sigma_m$ with $\sigma_i = (f(i),\vecj(i))$ representing the $i$th constraint.
The key complexity parameter is the space required to solve $\gbMaxF$.
A streaming algorithm is said to solve $\gbMaxF$ correctly if it outputs the correct answer with probability at least $2/3$ (i.e., it errs with probability at most $1/3$).

A sketching algorithm is a special class of a streaming algorithm, where the algorithm's output is determined by a small sketch it produces of the input stream, and the sketch itself has the property that the sketch of the concatenation of two streams can be composed from the sketches of the two component streams. A sketching algorithm is said to be a linear sketching algorithm if the composition function is linear. Their formal definitions can be found in \cite{CGSV24}.
\textcite{CGSV24} gave the following dichotomy characterization for the approximability of CSPs by small-space sketching algorithms.\footnote{
    Their dichotomy theorem also holds for \emph{dynamic} streaming algorithms.
    In this setting, constraints may be inserted, even multiple times, and later deleted, and algorithms are required to be correct on the final instance, under the promise that constraints were deleted fewer times than they were inserted at all intermediate stages of the streaming process.
    The input stream can be unboundedly large in this setting even while maintaining polynomially bounded integer weights (e.g., by inserting and deleting the same constraint an arbitrary number of times).}

First, we begin with some definitions.

\begin{definition}[Space of $\yes$/$\no$ distributions]\label{def:sysn}
For $q,k \in \N$, $\gamma \in [0,1]$ and $\cF\subseteq([q]^k)^{\{0,1\}}$,  let 
$$\sgyf \coloneqq \left\{\cD \in \Delta(\cF \times [q]^k) \mid \Exp_{(f,\ba) \sim \calD} [f(\ba)] \geq \gamma\right\}.$$
For $\beta \in [0,1]$, let 
$$\sbnf \coloneqq \left\{\cD \in \Delta(\cF \times [q]^k) \mid \Exp_{(f,\ba) \sim \calD} \bracks*{ \Exp_{\substack{\bb\in [q]^k, \\ b_\ell \sim \calP_{a_\ell}}} \bracks*{ f(\bb) } } \leq \beta\right\}.$$
\end{definition}

\begin{remark}
The definition of the $\yes$ and $\no$ distributions in \cite{CGSV24} is different than but equivalent to our \Cref{def:sysn}.
\end{remark}

Given a distribution $\calD \in \Delta(\calF \times [q]^k)$, the \emph{marginal vector} of $\calD$, denoted $\bmu(\calD)$,
is a $\calF \times [k] \times [q]$ array whose $(f,\ell,\sigma)$-th entry is $\Pr_{(f,\ba)\sim\calF}[a_\ell = \sigma]$.

\begin{definition}[Marginals of Yes/NO Distributions]\label{def:marginals}
For $\gamma,\beta \in [0,1]$ and $\cF\subseteq([q]^k)^{\{0,1\}}$,  we let 
$$\kgyf = \{\vecmu(\cD) \in \R^{|\cF| \times k \times q} \mid \cD \in \sgyf \}
\mbox{ and }
\kbnf = \{\vecmu(\cD) \in \R^{|\cF| \times k \times q} \mid \cD \in \sbnf \}.$$
\end{definition}

We are now ready to state the dichotomy theorem from \cite{CGSV24}.

\begin{theorem}[Dichotomy theorem for $o(\sqrt n)$-space sketching, \cite{CGSV24}]\label{thm:cgsv}
For every $q,k \in \N$, every family of functions $\cF \subseteq ([q]^k)^{\{0,1\}}$ and for every $0\leq\beta<\gamma\leq1$, the following hold:
\begin{enumerate}
    \item If $\kgyf \cap \kbnf = \emptyset$, then $(\gamma,\beta)$-$\maxF$ admits a linear sketching algorithm
    that uses $O(\log^3 n) $ space\footnote{In particular, the space complexity is $O(\log^3 n)$ bits, or $O(\log^2 n)$ words where each word is $O(\log n)$ bits long. Crucially while the constant in the $O(\cdot)$ depends on $k$, $\gamma$ and $\beta$, the exponent is a universal constant.} on instances on $n$ variables. \label{thmpart:positive-result-detailed}
    \item If $\kgyf \cap \kbnf \neq \emptyset$, then for every $\epsilon>0$, every sketching algorithm for $(\gamma-\epsilon,\beta+\epsilon)$-$\maxF$ 
    requires $\Omega(\sqrt{n})$ space\footnote{Again, the constant hidden in the $\Omega$ notation depends on $k$, $\gamma$ and $\beta$.} on instances on $n$ variables. Furthermore, if $\gamma = 1$, then every sketching algorithm for $(1,\beta+\epsilon)$-$\maxF$ requires $\Omega(\sqrt{n})$ space. 
    \label{thmpart:negative-result-detailed}
 \end{enumerate}
\end{theorem}

\begin{definition}[{\cite[Def.~3.11]{CGSV24}}]
    Let $q, k \in \N$.
    A predicate $f : [q]^k \to \{0,1\}$ \emph{supports one-wise independence} if there exists a distribution $\calD \in \Delta([q]^k)$ s.t.:
    \begin{enumerate}
        \item \emph{$\calD$ is supported on satisfying assignments of $f$:} $\Exp_{\bb \sim \calD} [f(\bb)] = 1$.
        \item \emph{$\calD$ has uniform marginals:}
        For every $\ell \in [k]$ and $\sigma \in [q]$, $\Pr_{\bb = (b_1,\ldots,b_k)\sim\calD} [b_\ell=\sigma] = \frac1q$.
    \end{enumerate}
    A family $\calF \subseteq ([q]^k)^{\{0,1\}}$ \emph{strongly supports one-wise independence}
    if every $f \in \calF$ supports one-wise independence,
    and \emph{weakly supports one-wise independence} if there exists a subfamily $\calF'\subseteq \calF$ s.t. $\rho(\calF) = \rho(\calF')$ which strongly supports one-wise independence.
\end{definition}

\begin{theorem}[{\cite[Thm.~3.12]{CGSV24}}]\label{thm:cgsv:one-wise}
    Let $q,k \in \N$ and $\calF \subseteq ([q]^k)^{\{0,1\}}$.
    If $\calF$ weakly supports one-wise independence, then $\maxF$ is approximation-resistant to $o(\sqrt n)$-space streaming algorithms,
    i.e., every algorithm for $\gapMaxF{1}{\rho(\calF)+\epsilon}$ requires $\Omega(\sqrt n)$ space.
\end{theorem}

We remark that it is not currently known whether there exists any predicate family weakly (or even strongly) supporting one-wise independence which is nontrivially approximable to $o(n)$-space streaming algorithms.
As in \Cref{rmk:ltf}, the \textcite{Pot19} balanced LTF seems to be a good candidate for this.

\subsection{Linear space streaming lower bounds}

The following theorem due to Chou, Golovnev, Sudan, Velingker, and Velusamy gives the current best linear space lower bounds for every Max-CSP over $\Z_q^k$ in the streaming setting \cite[Theorem 4.3]{CGS+22-linear-space}\footnote{Note that their result uses the fact that $\Z_q$ is an additive group.}.
\begin{definition}[Width of $\cF$]
	For $\vecb \in \Z_q^k$ and a predicate $f:\Z_q^k \to \{0,1\}$, the $\vecb$-width of $f$, denoted $\omega_{\vecb}(f)$, is the quantity $\omega_\vecb(f) \coloneqq \frac{|\{a\in\Z_q\, |\, f(\vecb + a^k)=1\}|}q$. The {\em width} of $f$, denoted $\omega(f)$, is given by $\omega(f) \coloneqq \max_{\vecb \in \Z_q^k} \{\omega_{\vecb}(f)\}$. The \emph{width} of a \emph{family} $\cF \subseteq (\Z_q^k)^{\{0,1\}}$ is $\omega(\cF) \coloneqq \min_{f\in\cF} \{\omega(f)\}$. A family $\cF$ is {\em wide} if $\omega(\cF) = 1$. 
\end{definition}

\begin{theorem}[\cite{CGS+22-linear-space}]\label{thm:cgsvv}
    For every $k,q, \cF \subseteq (\Z_q^k)^{\{0,1\}}$ and every $\epsilon > 0$, every single-pass streaming for $(\omega(\cF),\rho(\cF)+\epsilon)$-$\maxF$ requires $\Omega(n)$ space.
\end{theorem}

The above theorem proves optimal hardness results for some $\maxF$ problems such as $\MaxCut$, $\MaxDiCut$ and $\MaxkAnd$.

\subsection{LP approximations for CSPs}\label{sec:basic lp}

\begin{definition}[{The basic LP, \cite{GT17}}]\label{def:basic LP}
Let $\calF$ be a predicate family and $\Phi$ an instance of $\maxF$.
The basic LP corresponding to $\Phi$ is as follows.
For every constraint $C = (f,\bj) \in \Phi$, where $f \in \calF$ and $\bj = (j_1,\ldots,j_k)$ are distinct variables in $[n]$,
we create $q^k$ corresponding LP variables $(y_{C,\ba})_{\ba\in [q]^k}$,
and for every variable $i \in [n]$, we create $q$ corresponding LP variables $(x_{i,b})_{b \in [q]}$.
The LP is then:

\begin{empheq}[left=\empheqlbrace]{alignat*=3}
    & \mathrm{maximize} \quad && \Exp_{C = (f,\bj) \in \Phi} \sum_{\ba \in [q]^k} f(\ba) \cdot y_{C,\ba} && \\
    & \mathrm{s.t.} && \sum_{b \in [q]} x_{i,b} = 1, && \forall i \in [n], \\
    & && \sum_{\ba \in [q]^k : a_\ell = b} y_{C,\ba} = x_{i,b}, && \forall C = (f, \bj) \in \Phi, j_\ell = i, b \in [q], \\
    & && x_{i,b} \ge 0, && \forall i \in [n], b \in [q], \\
    & && y_{C,\ba} \ge 0, && \forall C \in \Phi, \ba \in [q]^k.
\end{empheq}
We let $\LPopt(\Phi)$ denote the corresponding value of the basic LP.
\end{definition}

For a constraint $C$, the collection of variables $(y_{C,\ba})_{\ba \in [q]^k}$ can be viewed as a distribution $Y_C$ on $[q]^k$, called a \emph{local distribution}.
The following alternative formulation of the basic LP is therefore easily seen to be equivalent to \Cref{def:basic LP}.

\begin{definition}[Distributional view of basic LP]\label{def:basic LP:dist}
    Optimize the following over the tuples of local distributions $(Y_C \in \Delta([q]^k))$:
    \begin{empheq}[left=\empheqlbrace]{alignat*=3}
    & \mathrm{maximize} \quad && \Exp_{C = (f,\bj) \in \Phi} \bracks*{ \Exp_{\ba \sim Y_C} f(\ba) } && \\
    & \mathrm{s.t.} && \Pr_{\ba \sim Y_C}[a_\ell = b] = \Pr_{\ba \sim Y_{C'}}[a_{\ell'} = b], \qquad && \forall C = (f, \bj), C' = (f, \bj') \in \Phi, j_\ell = j'_{\ell'}, b \in [q].
\end{empheq}
The condition implies that there exists a one-wise marginal distribution $X_i \in \Delta([q])$ for every $i \in [n]$,
such that for every constraint $C = (f,\bj) \in \Phi$ s.t. $j_\ell = i$,
the marginal distribution of $Y_C$ on the $\ell$-th coordinate is $X_i$.
\end{definition}
Note that these local distributions need not be marginals of a single global distribution
(and this is why the basic LP is a strict relaxation of the CSP),
but they do at least ensure that the one-wise marginals of the local distributions are consistent:
If there are two constraints $C = (f, \bj)$ and $C' = (f', \bj')$ both containing a variable $i \in [n]$, i.e., $j_\ell = j'_{\ell'} = i$,
then the local distributions $Y_C$ and $Y_{C'}$ should agree on the corresponding coordinates.

\begin{remark}
One can consider a stronger LP requiring consistency on higher-degree marginals; this is essentially the so-called ``Sherali-Adams'' hierarchy~\cite{SA90}.
Ghosh and Tulsiani~\cite{GT17} showed that integrality gaps for the basic LP extends to integrality gaps for the level-$\Omega(\frac{\log n}{\log\log n})$ Sherali-Adams relaxation.
By results of Chan, Lee, Raghavendra, and Steurer~\cite{CLRS16} and Kothari, Meka, and Raghavendra~\cite{KMR22},
these integrality gaps in turn imply bounds against polynomial-sized extended LP formulations.
\end{remark}

\section{Results}

\subsection{Lower bounds against streaming approximability}

\begin{theorem}[Detailed version of \Cref{thm:int gap impl CGSV}]
    Suppose there exists a $(\gamma,\beta)$-integrality gap instance for the basic LP for $\maxF$.
    Then $\kgyf \cap \kbnf \ne \emptyset$, and therefore, for every $\epsilon > 0$,
    every sketching algorithm for $(\gamma-\epsilon,\beta+\epsilon)$-$\maxF$ on $n$-variable instances uses $\Omega(\sqrt n)$ space.
\end{theorem}

\begin{proof}
Our proof is via a constructive transformation from an integrality gap instance to a pair of distributions with matching marginals.

Suppose we have a $(\gamma,\beta)$-integrality gap instance for $\maxF$.
That is, there is an instance $\Phi$ of $\maxF$ s.t. $\CSPopt_\Phi \le \beta$ but $\LPopt_\Phi \ge \gamma$,
and the latter is witnessed by a one-wise-consistent assignment of local distributions $(Y_C \in \Delta([q]^k))_{C \in \Phi}$.
Let $X_i$ denote the marginal distribution variable $i \in [n]$.

Unpacking the definitions, to show that $\kgyf \cap \kbnf \ne \emptyset$, we must show that there exist distributions $\calD^\yes,\calD^\no \in \Delta(\calF \times [q]^k)$
s.t. $\bmu(\calD^\yes) = \bmu(\calD^\no)$, $\calD^\yes \in \sgyf$, and $\calD^\no \in \sbnf$.
In turn, these latter conditions mean that \[
\Exp_{(f,\ba) \sim \calD^\yes} [f(\ba)] \ge \gamma
\quad\text{and}\quad \forall (\cP_\sigma \in \Delta([q]))_{\sigma \in [q]}, \Exp_{(f,\ba) \sim \calD^\no} \bracks*{ \Exp_{\substack{\bb\in [q]^k, \\ b_\ell \sim \calP_{a_\ell}}} [f(\bb)] } \leq \beta.
\]

Our construction is the following:
\begin{itemize}
\item We set $\calD^\yes$ to be the distribution where we sample $C = (f,\bj) \sim \Phi$,
sample $\ba \sim Y_C$ (where $Y_C$ is the local distribution for the constraint $C$),
and output the pair $(f,\ba)$.
\item We set $\calD^\no$ to be the distribution where we sample $C = (f,\bj) \sim \Phi$,
sample $\ba = (a_1,\ldots,a_k)$ by sampling each $a_i \sim X_{j_i}$ independently,
and output $(f,\ba)$.
\end{itemize}
To analyze this construction, we first observe that \[
\Exp_{(f,\ba)\sim\calD^\yes} [f(\ba)] = \Exp_{C=(f,\bj)\sim\Phi} \bracks*{ \Exp_{\ba \sim Y_C} [f(\ba)] } \ge \gamma, \]
where the equality uses the definition of $\calD^\yes$ and the inequality is the assumption on $\LPopt_\Phi$.
On the other hand, fix any distributions $\calP_\sigma \in \Delta([q])$ for $\sigma \in [q]$.
Then
\begin{multline*}
\Exp_{(f,\ba) \sim \calD^\no} \bracks*{ \Exp_{\substack{\bb\in [q]^k, \\ b_\ell \sim \calP_{a_\ell}}} [f(\bb)] } = \Exp_{(f,\bj) \sim \Phi} \bracks*{ \Exp_{\substack{\ba \in [q]^k, \\ a_\ell \sim X_{j_\ell}}} \bracks*{ \Exp_{\substack{\bb\in [q]^k, \\ b_\ell \sim \calP_{a_\ell}}} [f(\bb)] } } = \Exp_{(f,\bj) \sim \Phi} \bracks*{ \Exp_{\substack{\bx \in [q]^n, \\ x_i \sim X_i}} \bracks*{ \Exp_{\substack{\bz \in [q]^n, \\ z_i \sim \calP_{x_i}}} \bracks*{  f(z_{j_1},\ldots,z_{j_\ell}) } } } \\
= \Exp_{\substack{\bx \in [q]^n, \\ x_i \sim X_i}} \bracks*{ \Exp_{\substack{\bz \in [q]^n, \\ z_i \sim \calP_{x_i}}} \bracks*{ \Exp_{(f,\bj) \sim \Phi} [f(z_{j_1},\ldots,z_{j_\ell})] } }
= \Exp_{\substack{\bx \in [q]^n, \\ x_i \sim X_i}} \bracks*{ \Exp_{\substack{\bz \in [q]^n, \\ z_i \sim \calP_{x_i}}} \bracks*{ \CSPval_\Phi(\bz) } }
\le \beta,
\end{multline*}
where the equalities use, respectively, the definition of $\calD^\no$, the fact that the marginal distribution of $(z_{j_1},\ldots,z_{j_\ell})$ is precisely the marginal distribution of $\bb$,
rearranging expectations, and the definition of CSP value, and the inequality is the assumption on $\CSPopt$.
\end{proof}

\subsection{Lower bounds against LP approximability}

\begin{theorem}[Detailed version of \Cref{thm:one-wise}]\label{thm:one-wise:formal}
    Let $\calF \subseteq \{f : [q]^k \to \{0,1\}\}$.
    Suppose that $\calF$ weakly supports one-wise independence.
    Then for every $\epsilon > 0$, there exists $(1,\rho(\calF)+\epsilon)$-integrality gap instance for the basic LP for $\maxF$, i.e., $\maxF$ is LP-approximation resistant.
\end{theorem}

\begin{proof}
    We assume without loss of generality that $\calF$ strongly supports one-wise independence.
    The theorem follows from the following claim: \emph{Every instance $\Phi$ of $\maxF$ has $\LPopt_\Phi = 1$.}
    Given this, the theorem follows immediately from the existence of instances $\Phi$ with value $\CSPval_\Phi \le \rho(\calF) + \epsilon$ for every $\epsilon > 0$ (as in the definition of $\rho(\calF)$).
    Indeed, let $\Phi$ be any instance of $\maxF$.
    Then we can simply set the local distribution $Y_C$ for every constraint $C = (f,\bj) \in \Phi$ to the corresponding distribution $\calD_f$.
    This solution has objective value $1$, since every local distribution is supported only on satisfying assignments,
    and the local distributions have consistent marginals, since they are all uniform on $[q]$.
\end{proof}

\begin{theorem}[Detailed version of \Cref{thm:wide}]\label{thm:wide:formal}
     For every $k,q, \cF \subseteq (\Z_q^k)^{\{0,1\}}$ and every $\epsilon > 0$, there exists a $(\omega(\cF),\rho(\cF)+\epsilon)$-integrality gap instance for $\maxF$.
\end{theorem}

\begin{proof}
Like in the proof of \cref{thm:one-wise:formal}, we pick an instance $\Phi$ such that $\CSPval_\Phi \le \rho(\calF) + \epsilon$. For $f\in \cF$, let $\vecb_f\in \Z_q^k \triangleq \arg \max_{\vecb \in \Z_q^k} \omega_{\vecb}(f)$. Let $\calD_f$ denote the uniform distribution over $\{\vecb_f + a^k: a\in \Z_q\}$. Then we can set the local distribution $Y_C$ for every constraint $C = (f,\bj) \in \Phi$ to the corresponding distribution $\calD_f$. Observe that these local distributions have consistent marginals since they are all uniform on $\Z_q$. It follows from the definition of $\vecb_f\in \Z_q^k$ that $\Exp_{\ba \sim Y_C} f(\ba) = \omega(f)$. Hence, the solution has objective value at least $\omega(\cF)=\min_{f\in\cF} \{\omega(f)\}$.
\end{proof}

Thus, for every previously known linear space lower bounds for Max-CSPs in the streaming setting, we construct corresponding integrality gap instances for their basic LPs.

\printbibliography
\end{document}